%% file: edge.tex
\documentclass[cleveref,autoref,thm-restate]{article}

\usepackage{header}
\usepackage{tda}

\begin{document}
\thispagestyle{empty}
\maketitle

\begin{abstract}
    \input{body/abstract}
\end{abstract}

\section{Introduction}\label{sec:intro}
\input{body/intro}

\section{Background and Tools}\label{sec:tools}
\input{body/tools}

\section{Fast Reconstruction}\label{sec:edges}
\input{body/edge-reconstruction}

\section{Discussion}
\input{body/discussion}

\small
\bibliographystyle{abbrv}
\bibliography{references}

\newpage
\appendix
\section{Algorithmic Proofs}\label{append:proofs}
\input{body/append-proofs}

\section{Basis}\label{append:basis}
\input{body/append-basis}

\end{document}

%% file: body/abstract.tex
Persistent homology is a tool that can be employed to summarize the \emph{shape}
of data by quantifying homological features.  When the data is an object in
$\R^d$, the (augmented) persistent homology transform~((A)PHT) is a
family of persistence diagrams, parameterized by directions in the ambient
space.  A recent advance in understanding the PHT used the framework of
reconstruction in order to find finite a set of directions to faithfully
represent the shape, a result that is of both theoretical and practical
interest. In this paper, we improve upon this result and
present an improved algorithm for graph---and, more generally
one-skeleton---reconstruction.  The improvement comes in reconstructing the
edges, where we use a radial
binary (multi-)search.  The binary search
employed takes advantage of the fact that the edges
can be ordered radially with respect to a reference plane, a feature
unique to graphs.

%% file: body/intro.tex
At the heart of inverse problems in the field of topological data analysis is
the following question: \emph{how many persistence diagrams are needed to 
faithfully represent a shape?}  Since the introduction of persistence diagrams,
it has been known that many ``shapes'' can share the same persistence diagram.
With enough persistence diagrams, we arrive at a set of parameterized
diagrams (that is, diagrams labeled by direction) that uniquely, or
\emph{faithfully} represents the underlying shape.
Moreover, the set of parameterized diagrams is \emph{faithful} if and only if it can be used
to \emph{reconstruct} the underlying shape.

The foundation for asking the question of how many diagrams are needed
for a faithful representation
was first introduced
in~\cite{turner2014persistent}, where Turner et al.\ defined the Persistent Homology
Transform (PHT) and the Euler characteristic curve transform (ECCT). These
transforms map a simplicial complex embedded in~$\R^{d}$ (a shape) to sets of
persistence diagrams (respectively, Euler characteristic curves) 
parmeterized by $\S^{d-1}$, the set of all directions in~$\R^d$. 
They showed that, up to mild general position assumptions, no two
simplicial complexes can correspond to the same PHT or ECCT, i.e., 
they showed that the uncountably infinite sets making up the PHT and ECCT
faithfully represent the~shape.

However, the PHT and ECCT are uncountably
infinite sets of diagrams (and Euler characteristic curves).  Thus, to bridge the
gap between the theory and what can be used in practice, a discretization of the
PHT was needed; several papers stepped up to the challenge and proved that
there exists a finite discretizations of topological transforms that are faithful for
simplicial and cubical
complexes~\cite{belton2018learning,belton2019reconstructing,betthauser2018topological,%
ghrist2018euler, curry2018directions,micka2020searching, fasy2022faithful}.
Beyond proving the existence of these finite faithful sets, Belton et
al.~\cite{belton2019reconstructing}
explicitly give an algorithm for using an oracle to reconstruct graphs embedded
in $\R^d$ with $n$ vertices. 
Their reconstruction uses~$n^2 - n + d + 1$ augmented persistence diagrams
in~$O(dn^{d+1}+ n^4 + (d + n^2)\oracleTime)$ time~\cite[Theorem~17]{belton2019reconstructing},
where~$\oracleTimeFull$ is the time complexity it
takes for an oracle to produce answer a persistence diagram
querry. Since the
direction-labeled set of diagrams can be used to reconstruct the underlying
graph, the set is a
faithful discretization of the augmented PHT~(APHT). 

In the current work, we give a faithful discretization of the APHT using
$\fullDGMCompFull$ diagrams and~$\fullTimeFull$ time.
This result is an
improvement over~\cite{belton2019reconstructing}, both in the size of the set and
in the speed of reconstruction.
The crux is in the improvmeent to edge reconstruction. While the method of \cite{belton2019reconstructing}
uses a linear scan of all possible edges for each vertex, resulting in a
quadratic number of diagrams needed, here we show that we can detect edges
with $\edgeDGMCompFull$ diagrams using
a radial binary multi-search.

%% file: body/tools.tex
In this section, we provide definitions of the tools used in the remainder of the
paper. We make use of standard notation such as using $e_i$ for the $i$th
standard basis vector in $\R^d$, where $1 \leq i \leq d$.
We use the notation~$(V, E)$ for a graph and its vertex and edge sets, and use $n = |V|$ and $m = |E|$.
We assume the reader is familar with standard concepts in topology (such as
simplicial
complexes, simplicial homology, Betti numbers), and note that further information can
be found in~\cite{edelsbrunner2010computational} for a general introduction and
in \cite[Section 2.1]{fasy2022faithful} for a detailed definition of the
augmented persistence diagram.

First, we define our general postition assumption.

\begin{assumption}\label{ass:general}
    Let $V \subset \R^d$ be a finite set with $d \geq 2$. We say $V$ is in
    \emph{general position} if the following properties are satisfied:
    \begin{enumerate}[(i)]
        \item Every set of $d+1$ points is affinely
            independent.\label{assitm:affine}
        \item No three points are colinear after orthogonal projection
            into the space~$\pi(\R^d)$,
            where $\pi: \R^d \to \R^2$ is the orthogonal projection onto
            the plane spanned by the first two basis elements, $e_1$ and~$e_2$.
            \label{assitm:projectedindep}
        \item Every point has a unique height with respect to the direction
            $e_2$.
            \label{assitm:unique}
    \end{enumerate}
    We call a graph \emph{GP-immersed}\footnote{Here, we use \emph{immersed}
    rather than \emph{embedded} in order to allow intersections of edges.  Note,
    however,
    that this can only happen when $d=2$.} iff its vertex set is in general
    position in~$\R^d$.
\end{assumption}

We note that \ref{assitm:unique} is not strictly necessary, however, it is
convenient for simplicity of exposition.  How to handle this degeneracy is discussed in
\appendref{basis}.

Given a graph GP-immersed in $\R^d$, we can filter the graph based on the height
in any direction $\dir$ in the sphere of directions $\S^{d-1}$.  To do so, we
assign each simplex a height.  A vertex $v \in V$ is assigned the height~$s
\cdot v$, and an edge $[v_0,v_1] \in E$ is assigned the height $\max \{s\cdot
v_0, s \cdot v_1\}$.  This function, mapping vertices and edges to heights, is
known as a \emph{filter} function, which we use to compute persistent homology.

\subsection{Persistence and the Oracle Framework}

Given a filtered simplicial complex (that is, a simplicial complex with each
simplex assigned a ``height''), the corresponding
\emph{augmented persistence diagram~(APD)} is a record of
of all homological events throughout the filtration.
A \emph{birth event} is the introduction or appearance of a new homological feature,
and a \emph{death event} is the merging of two features.  A death is paired with
the most recent of the birth-labeled features that it merges together,
creating a birth-death
pairing, leaving the remaining feature labeled by the elder birth height.
In an APD,
every simplex corresponds to exactly one event (resulting in some pairings where
the birth and death heights are equal).
As a result,
by construction, APDs contain at least one event at the height of each vertex.

For a simplicial complex (e.g., a graph) GP-immersed in $\R^d$ and a
direction in $S^{d-1}$, we use the \emph{lower-star filtration}:
the nested sequence of  graphs that arise by
looking at all simplices at or below a given height and allowing that height to
grow from $-\infty$ to~$\infty$. Throughout this paper, we denote
the~$i$-dimensional APD by~$\dgm{i}{\dir}$ and write $\dgm{}{\dir}=\sqcup_{i}
\dgm{i}{\dir} $, omitting the graph from the notation (as it is always
clear from~context).\footnote{
When calculating diagrams, we count $\dgm{}{\dir}$ as one diagram, not
multiple.}

The \emph{(augmented) persistence homology transform ((A)PHT)} is the set of
(augmented)
diagrams of lower-star filtrations in all possible directions, parameterized by
direction.
That is, the set $X=\{(\dgm{}{\dir}, \dir)\}_{\dir \in \sph^{d-1}}$. A
\emph{faithful discretization} is a finite subset of $X$
from which all other elements of $X$ can be deduced (and, by
\cite{turner2014persistent}, corresponds to a
unique simplicial complex).
The introduction of the (augmented) persistent homology transform has
sparked related research in applications of shape
comparison~\cite{giusti2015clique, lee2017quantifying, rizvi2017single,
lawson2019persistent, tymochko2020using,
wang2019statistical,singh2007topological,bendich2016persistent}.
As such, finding a minimal faithful discretization is important for the
applicability of the (A)PHT.
In what follows, we will only consider APDs, and we may shorten notation and
refer to an APD by the word \emph{diagram}.

In this work, we assume an oracle framework.  That is, we assume that we have no
knowledge of the shape itself, but we have access to an oracle from which
we can query directional diagrams.

\begin{definition}[Oracle]
\label{def:oracle}
    For a graph $(V,E)$ GP-immersed in $\R^d$ and
    a direction~$\dir \in \sph^{d-1}$,
    the operation~$\oracle{\dir}{}$ returns the diagram~$\dgm{}{\dir}$.
    We define~$\oracleTimeFull$ to be the time complexity of this oracle query and
    note that the space complexity of $\dgm{}{\dir}$ is~$\Theta(n+m)$.
    We assume that the data structure returned by the oracle allows queries for
    specific birth or death values in~$\Theta(\log n)$ time (for example, the
    we could have two arrays of persistences points, one sorted by birth values
    and one sorted by death values).
\end{definition}

\subsection{Constructions and Data Structures}
\label{ss:constructions}
In this subsection, we introduce the edge arc object and other definitions
useful for computing properties of immersed graphs.
Throughout this paper, we project points in $\R^d$ to the~$(e_1,e_2)$-plane.
As a result, we use
``above (below)'' without stating with respect to which direction
as shorthand for ``above (below) with respect to the
direction~$e_2$.''  This direction is intentionally chosen (and used in our GP
assumption), as it corresponds to
our intuition of above (below) in the figures.
When we measure an angle of a vector $x$, denoted~$\measuredangle
x$, we mean the angle
that $\pi(x)$ makes with the positive $e_1$ axis.

Given a direction $\dir$ and a vertex in a graph immersed in~$\R^d$, we classify
each edge $(v,v')$ as either an ``incoming'' edge, when $v'$ is below $v$ with respect to $\dir$,
or an ``outgoing'' edge, when $v'$ is
above $v$ with respect to $\dir$. Note
that all incoming edges have the same height as the vertex with respect to the
$e_2$ direction.

\begin{definition}[Indegree]\label{def:indegree}
    Let $(V,E)$ be a graph GP-immersed in $\R^d$. Let $v \in V$ and $\dir \in
    \S^{d-1}$. The
    \emph{indegree of~$v$ in direction $\dir$}, denoted $\indeg{v}{\dir}$, is
    the number of edges incident to $v$ with height $\dir \cdot v$.
\end{definition}

The following lemma relates
the number of edges at a given height to points in the APD.
\begin{lemma}[Edge Count]\label{lem:count}
    Let $(V,E)$ be a graph and let $c \in \R$.
    Let~$f \colon V\sqcup E \to \R$ be a filter function.
    Then, the edges in $E$ with a function value of
    $c$ are in one-to-one
    correspondence with the following multiset of points in $\dgm{}{f}$, the diagram
    corresponding to $f$:
    \begin{align}
        \begin{split}
            \big\{ (b,d) &\in \dgm{1}{f} \text{ s.t. } b = c \big\} \label{eqn:isimps} \\
            &\cup \big\{ (b,d) \in \dgm{0}{f} \text{ s.t. } d = c \big\}.
        \end{split}
    \end{align}
\end{lemma}
In other words, each edge corresponds to either a birth of a one-dimensional
homological
feature or a death of a zero-dimensional feature in $\dgm{}{f}$.
For more details and a generalized proof, see~\cite[Appendix~A]{fasy2022faithful}.

If $f$ is a lower-star filtration in direction $\dir \in \S^{d-1}$, we
note that whenever a vertex $v$ has a unique height with respect to a direction
$\dir$, the cardinality of the multiset above is exactly~$\indeg{v}{\dir}$.

\begin{lemma}[Indegree Computation]\label{lem:indegree}
    Let~$(V,E)$~be a graph GP-immersed in~$\R^d$. Let~$v \in V$
    and let~$\dir \in \sph^{d-1}$ such that $\dir \cdot v \neq \dir \cdot v'$
    for any~$v' \neq v \in V$. Then,~$\indeg{v}{\dir}$ can be computed via the
    oracle using
    \zeroindegDGMComp diagram and~$\zeroindegTimeFull$ time.
\end{lemma}
\begin{proof}
    Let~$\adgm=\oracle{\dir}{}$.
    By the assumption on $s$,
    the height of $v$ with respect to the direction $\dir$ is unique.
    Hence, we
    know that any edge at height~$c=\dir \cdot v$ must be incident to $v$.
    Thus, by the definition of  indegree, an edge
    has the height~$c$ if and only if it contributes to the indegree of $v$
    in direction~$\dir$.
    Using \lemref{count}, we count these edges by counting
    one-dimensional births and zero-dimensional deaths at
    height~$c$. Since $\adgm_0$ and~$\adgm_1$ are sorted by
    both birth and death values (see \defref{oracle}) and
    since $\adgm$ has~$\Theta(n+m)$ points, searching for these
    events takes~$\Theta(\log n + \log m)$.  Adding $\oracleTimeFull$ for the
    oracle query and recalling that $m=O(n)$, the total runtime
    is~$\zeroindegTimeFull$.
\end{proof}

\begin{table}[t!b!h]
    \caption{Attributes of the edge arc object.}\label{tab:edgearc}
    \centering
    \begin{tabular}{|r|p{0.3\textwidth}|}
	\hline
        \rowcolor{lightgray}
        $\eAVar$ & \textbf{Edge Arc}\\
	\hline
        $\thevert$     &   Vertex around which the edge arc is centered\\
	\hline
        $(\alpha_1,\alpha_2)$     &   Start and stop angles of the arc,
		with respect to the $e_1$ direction\\
	\hline
        $\vV$   &   Array of vertices in arc radially ordered clockwise in
            $(e_1, e_2)$-plane\\
	\hline
        $\eC$   &   Number of edges incident to $v$ within the arc\\
	\hline
    \end{tabular}
\end{table}
We conclude this section by introducing a data structure, the \emph{edge arc object}; see
\tabref{edgearc} for a summary of the attributes of an edge arc and \figref{edgealg-dsex} for an example.
An edge arc represents the region in the~$(e_1, e_2)$-plane centered at $v$ that
is
swept out between the two angles~$\alpha_1$ and~$\alpha_2$ (the word `arc' is
referring to the arc of angles between $\alpha_1$ and $\alpha_2$, where the
angle is measured with respect to the postive $e_1$ axis).
We only consider edge arcs in the upper
half-space, with respect to the $e_2$ direction,
so the maximal edge arc is the upper half-plane
and the start and stop angles always satisfy~$0 \leq \alpha_1 \leq \alpha_2 \leq \pi$.
An edge arc stores an array of vertices sorted radially clockwise about $\pi(v)$
in the $(e_1,e_2)$-plane in decreasing
angle with the $e_1$-direction.
By construction, the first vertex in the array must be closest to $\alpha_2$ and
the last closest
to~$\alpha_1$. The edge arc also stores the count of edges of~$E$ that have vertices
from~$\vV$ as endpoints.
In implementation, angles~$\alpha_1$ and $\alpha_2$
do not need to be stored directly, but we include them in psuedocode and
discussions for~clarity.

\begin{figure}[tb]
    \center
    \includegraphics[width=.26\textwidth]{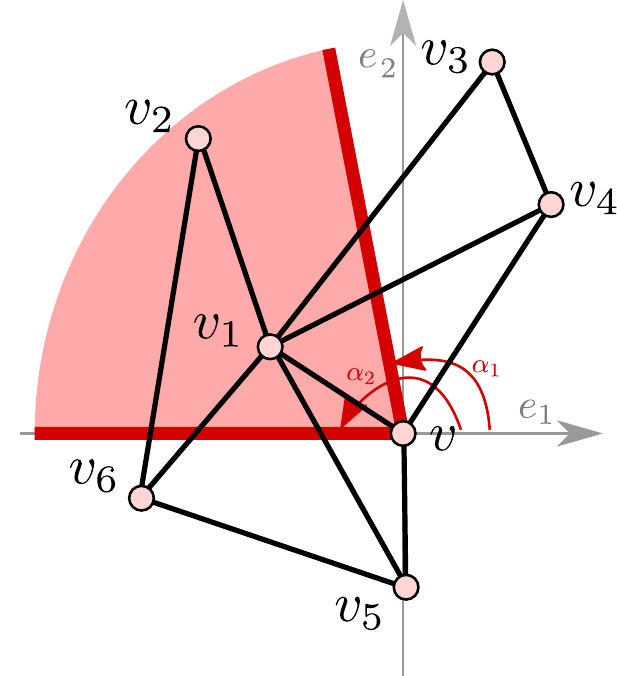}
    \caption
    {An edge arc $\eAVar$ centered at vertex~\mbox{$\eAVar.\thevert=v$}.  Other
    attributes of the edge arc include its start and stop angles,
    $\eAVar.\alpha_1 = 1.75$ radians and $\eAVar.\alpha_2 = \pi \approx 3.14$
    radians, the array of vertices~$\eAVar.\vV = \{v_1, v_2\}$,
    and the count of edges~$\eAVar.\eC=1$. Here, we also see that~$\indeg{v}{e_1}=1$
    and~$\indeg{v}{e_2}=2$.
    }\label{fig:edgealg-dsex}
\end{figure}

Given some arc $\eAVar$ centered at vertex $v \in V$, we need to
be able to compute $\eAVar.\eC$, the number of edges
contained $\eAVar$ that are adjacent to $v$.
The following lemma provides such a computation.
We omit a proof because it is a straightforward adaptation of \cite[Theorem
$16$]{fasy2022faithful} and~\cite[Lemma $13$]{belton2019reconstructing}.

\begin{lemma}[Arc Count]\label{lem:arccount}
    Let $(V,E)$ be a graph GP-immersed in $\R^d$. Let $\eAVar$ be an edge arc
    object, and let~$v=\eAVar.\thevert$. Let $\dir \in \sph^{d-1}$ be the direction
    perpendicular to~$\alpha_2$ so that the arc is entirely below
    $\dir \cdot v$. Let~$E_*$ denote the edges with
    height $\dir \cdot v$ that are not contained in $\eAVar$.
    If no vertex in $V$ is at the same height as~$v$ in direction $\dir$,
    then
    \begin{align*}
        \eAVar.\eC = \indeg{v}{\dir} - |E_*|.
    \end{align*}
\end{lemma}

As an illustration, again, consider \figref{edgealg-dsex}.  Consider the
direction $\dir = e^{i(\alpha_2-\pi/2)}$, which is perpendicular
to~$e^{i\alpha}$ by construction.  In addition, the edge arc is below~$\dir \cdot
v$ (specifically, all vertices in $\eAVar.\vV$ are below~$\dir \cdot v$).
Then, $ \indeg{v}{\dir}=2$ and $E_*=\{ [v,v_5]\}$.  By \lemref{arccount},~$\eAVar.\eC =
\indeg{v}{\dir} - |E_*|=2-1=1$.  When we say that a list of vertices or edges is
sorted clockwise around $v$, we mean that the list is sorted clockwise (cw)
around $\pi(v)$ once projected into the $(e_1,e_2)$-plane with the largest angle
first.

%% file: body/edge-reconstruction.tex
In this section, we provide an algorithm to reconstruct
a graph (and, more generally, a one-skeleton of a simplicial complex) using the
oracle.  We start
with an algorithm to find the edges, provided the
vertex locations are known.
We end with describing the complete graph reconstruction method.

\subsection{Fast Edge Reconstruction}\label{sec:fastedge}

In this subsection, we assume we have a graph $(V,E)$, where the vertex set~$V$ is known, but
$E$ is unknown.  Using the oracle and the known vertex locations, we provide a
reconstruction algorithm to find all edges in~$E$ (\algref{find-edges}).  This
algorithm is a sweepline algorithm in direction~$e_2$ that, for each vertex
processed in the sweep, performs a radial binary multi-search
(\algref{up-edges}).  This search is enabled by an algorithm that splits an edge
arc object into two edge arcs, each containing half of the vertices
(\algref{split-arc}).  We provide the algorithms and relevant theorem statements
here, but defer the proofs to \appendref{proofs}.


\input{body/algorithms/split-arc}
%
In \algref{split-arc}, we find a direction~$\dir$
in the~$(e_1,e_2)$-plane such that half of the
vertices in~$\eAVar.\vV$ are above~$v$ and half are below~$v$ with respect to
the direction~$\dir$. This allows us to create a new edge arcs corresponding to
each half; see \figref{EA}.
The properties of
\algref{split-arc} are described in the
following~theorem.

\begin{figure}[tb]
    \center
    \includegraphics[width=.39\textwidth]{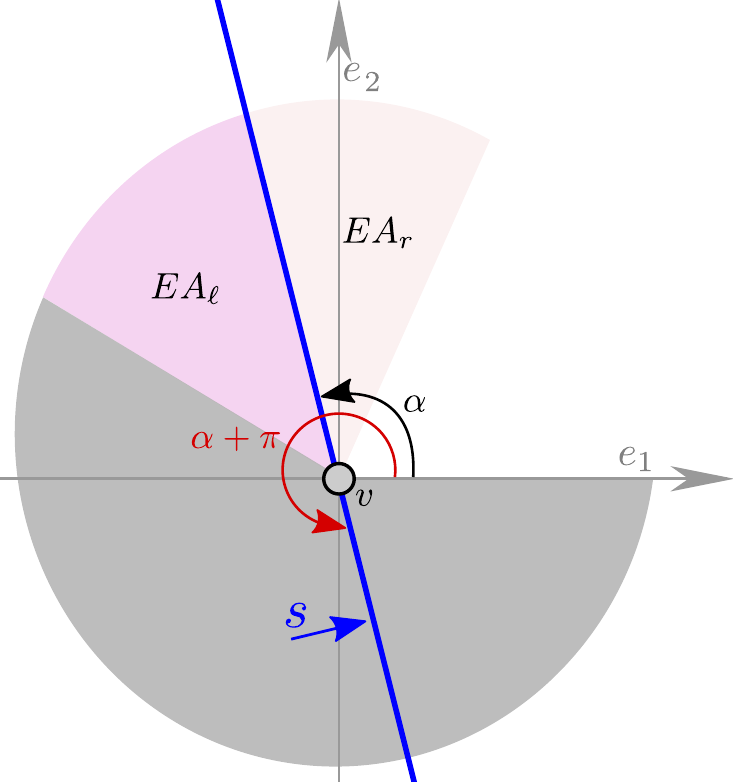}
    \caption
    {The splitting of edge arc $\eAVar$ into~$\eAVar_{\ell}$
    and~$\eAVar_r$, as in \algref{split-arc}. The large gray region
    is the region containing all edges of
    $bigedges$.  That is, all edges whose angle with the positive $e_1$-axis is
    at least~$\eAVar.\alpha_1$. On \alglnref{split:countell} of the algorithm,
    we compute the number of edges in
    $\eAVar_{\ell}$ by first computing the indegree of $\eAVar.\thevert$ in
    direction $s$ from the diagram in direction $s$, then we subtract the number
    of edges in
    $bigedges$ that are below the height of $\eAVar.\thevert$ in
    direction~$\dir$ (i.e., below the blue line).
    By the pigeonhole principal, we find~$\eAVar_r.\eC = \eAVar.\eC-\eAVar_{\ell}.\eC$.} \label{fig:EA}
\end{figure}

\begin{restatable}[Arc Splitting]{theorem}{splitarc}\label{thm:split-arc}
    \algref{split-arc} uses \splitDGMComp
    diagram and $\splitTimeFull$ time to split~$\eAVar$ into two new edge
    arcs~$\eAVar_{\ell}$ and~$\eAVar_r$ with the properties:
    \begin{enumerate}[(i)]
        \item \label{stmt:split-concat}
            The sets $\eAVar_r.\vV$ and $\eAVar_{\ell}.\vV$ partition
            the vertex set~$\eAVar.\vV$ such that the vertices
            in~$\eAVar_{\ell}.\vV$ come before those
            in~$\eAVar_r.\vV$, with respect to the clockwise ordering around
            $\eAVar.\thevert$.

        \item
            $|\eAVar_{\ell}.\vV| = \lceil{\frac{1}{2}|\eAVar.\vV|}\rceil$.
            \label{stmt:split-size-ell}

        \item
            $|\eAVar_{r}.\vV| = \lfloor{\frac{1}{2}|\eAVar.\vV|}\rfloor$.
            \label{stmt:split-size-r}
    \end{enumerate}
\end{restatable}


\input{body/algorithms/up-edges}

In \algref{up-edges}, we use \algref{split-arc} to find all outgoing edges from
a given vertex.  In particular, the
algorithm maintains a stack of edge arc objects. When processing an edge arc
(the while loop in \alglnrefRange{up-edges:loop}{up-edges:loopend}), we are
determining which of the vertices in $\vV$ form edges with~$\thevert$.
If an edge arc has
$\eC=0$, it contains no edges, and it can be ignored
(\alglnrefRange{up-edges:zero-edges}{up-edges:zero-edges-end}). If it
has $\eC$ exactly equal to the number of vertices in $\vV$, each vertex in $\vV$ must form
an edge with $v$
(\alglnrefRange{up-edges:one-arc}{up-edges:endone-arc}).  Otherwise, as
demonstrated in \figref{edgeAlg}, the edge arc is split in half using
\algref{split-arc} and each half is put on the stack to be processed.

\begin{figure*}[tbh]
    \center
    \begin{subfigure}{.33\textwidth}
        \centering
        \includegraphics[width=.8\textwidth]{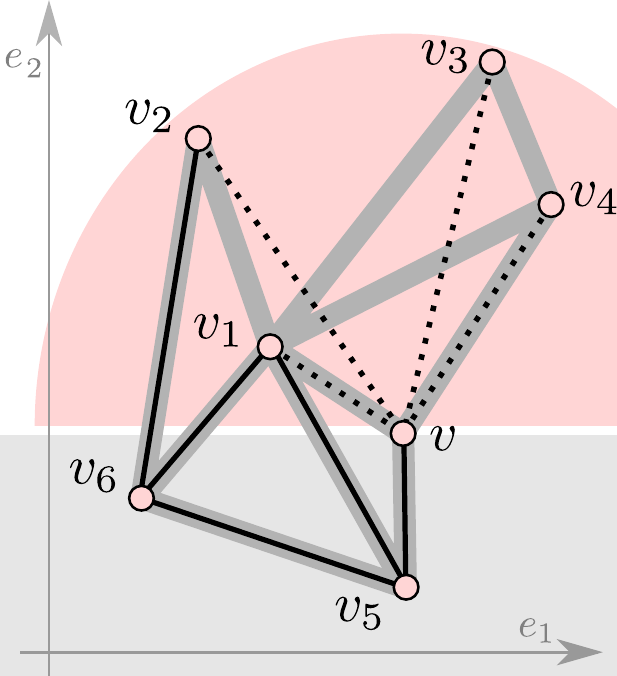}
        \caption{At vertex $v$.}
        \label{fig:edgeAlg-start}
    \end{subfigure}%
    \begin{subfigure}{.33\textwidth}
        \centering
        \includegraphics[width=.8\textwidth]{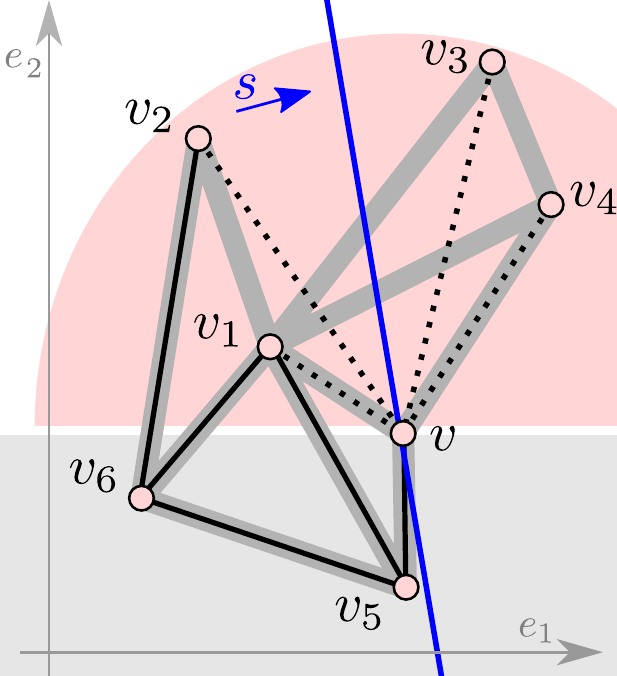}
        \caption{Split edge arc.}
        \label{fig:edgeAlg-ell}
    \end{subfigure}
    \begin{subfigure}{.33\textwidth}
        \centering
        \includegraphics[width=.8\textwidth]{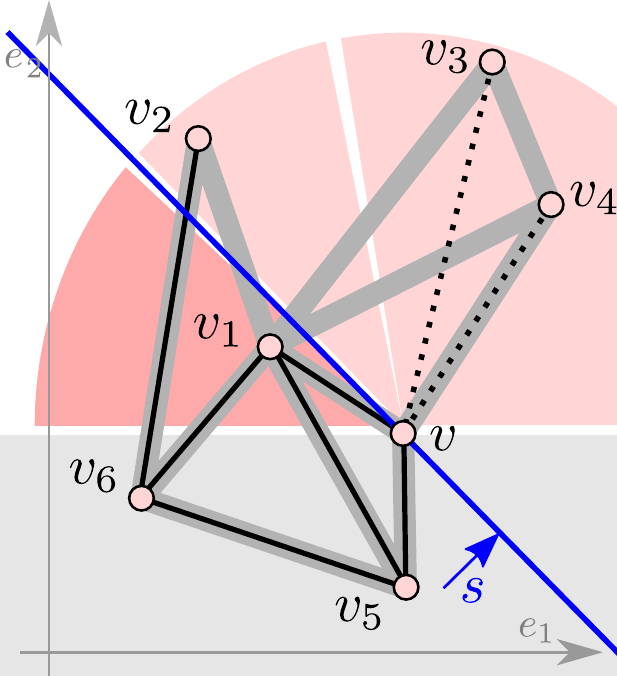}
        \caption{Split again.}
        \label{fig:edgeAlg-next}
    \end{subfigure}
    \caption[One step of \algref{up-edges}]
    {We demonstrate one step of \algref{up-edges}.  (\subref{fig:edgeAlg-start})
    By assumption, we initially know $[v_5,v]\in E$. From
    \alglnref{up-edges:indeg} of \algref{up-edges} we also know that two of the
    four vertices above $v$ are adjacent to $v$.  Thus, we create an edge arc
    object $\eAVar$ with $\eAVar.\eC=2$, and~$\eAVar.\vV=(v_1,v_2,v_3,v_4)$.
    (\subref{fig:edgeAlg-ell}) In \algref{split-arc}, we choose a direction
    $\dir$ such that half of the vertices in $\eAVar$ are below $v$. We use this
    split to create two edge arcs, $\eAVar_{r}$ and $\eAVar_\ell$, corresponding
    to the pink shaded regions on the right and left of the blue line defined by
    $\dir$.  We push~$\eAVar_{r}$ onto a stack to be processed later and focus
    on the arc $\eAVar_{\ell}$. Since two edges contribute to $v$'s indegree in
    direction~$\dir$ and one is the known edge~$[v_5,v]$, we have
    $\eAVar.\eC=2-1=1$. (\subref{fig:edgeAlg-next}) Next, we find a new
    direction~$\dir$ that splits~$\eAVar_{\ell}.\vV$ into two sets of size one.
    We push the set above $\dir$ onto our stack. The edge arc containing only
    $v_1$ also has $\eAVar.\eC=2-1=1$, so $[v_1, v] \in E$. After all steps of
    \algref{up-edges} are applied to find the edges above a particular vertex,
    \algref{up-edges} is then applied to the next highest vertex, eventually
    processing every vertex in $V$ in a sweep (\algref{find-edges}).

    }\label{fig:edgeAlg}
\end{figure*}

\begin{restatable}[Finding Edges Above a Vertex]{theorem}{upedge}\label{thm:up-edges}
    \algref{up-edges} finds the sorted array of edges above $v$
    using~$\upDGMCompFull$
    augmented persistence diagrams in~$\upTimeFull$~time.
\end{restatable}

Finally, our main algorithm (\algref{find-edges}) is a sweepline algorithm, where we
consider the vertices in increasing order of their $e_2$-coordinates and find
the outgoing edges of the vertex being considered.
\input{body/algorithms/find-edges}

\begin{restatable}[Edge Reconstruction]{theorem}{edgerec}\label{thm:edgerec}\label{thm:edges}
    Let $(V,E)$ be a graph GP-immersed in $\R^d$.
    Given~$V$, \algref{find-edges}
    reconstructs~$E$ using $\edgeDGMCompFull$ augmented persistence
    diagrams in~$\edgeTimeFull$~time.
\end{restatable}

\subsection{Putting it Together: Full Reconstruction}

The results of \secref{fastedge} are related to just part of
the full process of reconstruction, since reconstruction begins with no
knowledge of the underlying simplicial complex.  Identifying the location of all
vertices is the first step, and is one that has been previously
examined in detail. In particular, Belton et al.\
provide an algorithm
to reconstruct~$V$ in~$\vertTimeFull$ time and~$d+1$ oracle queries, assuming
stricter general position assumptions; see \cite[Algorithm 1 \& Theorem
9]{belton2019reconstructing}. 
An appropriate vertex reconstruction method taken together with \thmref{edges},
gives us the following for a full reconstruction~process.

\begin{theorem}[Graph Reconstruction]\label{thm:full}
    Suppose we use a vertex reconstruction method that requires
    $\vertTimeFullGeneral$ time and $\vertDGMFullGeneral$ diagrams. Then we can
    reconstruct an unknown graph immersed in~$\R^d$ using $\fullDGMCompFull$
    diagrams in~$\fullTimeFull$~time.
\end{theorem}

We omit a proof of \thmref{full}, as it simply combines the results of vertex
reconstruction methods (such as
\cite[Theorem 9]{belton2019reconstructing}) and \thmref{edges} of the current paper.
Note that, taking the stricter general position assumptions
of~\cite{belton2019reconstructing}, \thmref{full} gives us a runtime of
$\Theta(\vertTime + n^2 + m\log n(\splitTime))$ with~$\Theta(d + m \log
n)$ oracle queries.  

Observing that the methods presented here are immediately
applicable in the reconstruction of one-skeletons of general simplicial
complexes, we have the following corollary:

\begin{corollary}[One-Skeleton Reconstruction]
    Let~$K$ be an unknown simplicial complex GP-immersed in~$\R^d$.
    A vertex reconstruction algorithm that requires
    $\vertTimeFullGeneral$ time and $\vertDGMFullGeneral$ diagrams taken
    together with \algref{find-edges} of the current paper
    reconstruct the one-skeleton of
    $K$ using $\fullDGMCompFull$ augmented persistence diagrams
    in~$\fullTimeFull$~time.
\end{corollary}

Finally, we note that embedding a graph (or simplicial complex) in $\R^2$ is a
special case, as~$m=O(n)$ and~$d$
is constant. This speeds up edge reconstruction as well as potentially, vertex
reconstruction. For instance, with the stricter general position assumptions,
the method of
\cite[Theorem 6]{belton2019reconstructing} gives a vertex
reconstruction of a graph embedded in~$\R^2$ that uses three diagrams and
$\Theta(n\log n + \Pi)$ time.  Broadening our view again to general vertex
reconstruction algorithms, we obtain a result for plane graph
reconstruction:

\begin{corollary}[Reconstruction in $\R^2$]
    Suppose we use a vertex reconstruction method that requires
    $\vertTimeFullGeneral$ time and $\vertDGMFullGeneral$ diagrams.
    We can use an
    oracle to reconstruct the one-skeleton of an unknown simplicial complex
    embedded in $\R^2$
    using~$O(\vertDGMGeneral + n \log n)$ diagrams 
    and~$O(\vertTimeGeneral + n^2 + n \log n(d + \oracleTime))$~time. 
\end{corollary}

%% file: body/algorithms/split-arc.tex
\begin{algorithm}[tbh]\caption{$\splitArc{\eAVar}{\varbigedges}{\theta}$}\label{alg:split-arc}
    \begin{algorithmic}[1]
        \REQUIRE
            $\eAVar$, an edge arc;
            $\varbigedges$, an array of all edges~$(\eAVar.\thevert,v') \in E$ such
            that~$\measuredangle \pi(v'-\eAVar.\thevert) < \eAVar.\alpha_1$;
            $\theta$, the
            minimum angle defined by any three vertices in $\pi(V)$
        \ENSURE $\eAVar_{\ell}$ and $\eAVar_r$, edge arcs satisfying
            the properties in \thmref{split-arc}
        \STATE $n_v \gets |\eAVar.\vV|$\label{algln:split:numverts}
        \STATE $mid \gets \lceil \frac{n_v}{2} \rceil$
            \label{algln:split:calcmid}
        \STATE $\alpha \gets
            \measuredangle \pi(\eAVar.\vV[mid]-\eAVar.\thevert) - \theta/2$
            \label{algln:split:alpha}
        \STATE $\dir \gets e^{i(\alpha - \frac{\pi}{2})}$\label{algln:split:dir}
        \STATE $m_{\ell} \gets \indeg{\eAVar.\thevert}{\dir} - | \{ b \in bigedges ~|~
            \measuredangle b < \pi+\alpha \}|$\label{algln:split:countell}
        \STATE $m_{r} \gets \eAVar.\eC-m_{\ell}$\label{algln:split:countr}
        \STATE $\eAVar_{\ell} \gets$ edge arc
            where~$\eAVar_{\ell}.\thevert = \eAVar.\thevert$,
            $\eAVar_{\ell}.\alpha_1=\alpha$,
            $\eAVar_{\ell}.\alpha_2=\eAVar.\alpha_2$,
            $\eAVar_{\ell}.\vV=\eAVar.\vV[:mid]$,
            and~$\eAVar_{\ell}.\eC= m_{\ell}$
            \label{algln:split:ell}
        \STATE $\eAVar_{r} \gets$ edge arc
            where~$\eAVar_{r}.\thevert = \eAVar.\thevert$,
            $\eAVar_{r}.\alpha_1=\eAVar.\alpha_1$,
            $\eAVar_{r}.\alpha_2=\alpha$,
            $\eAVar_{r}.\vV=\eAVar.\vV[mid+1:]$,
            and~$\eAVar_{r}.\eC=m_r$
            \label{algln:split:r}
        \RETURN $(\eAVar_{\ell}$, $\eAVar_{r})$
    \end{algorithmic}
\end{algorithm}

%% file: body/algorithms/up-edges.tex
\begin{algorithm}[tbh]
    \caption{$\findUpEdges{v}{V_v}{in_v}{\theta}{\adgm}$ }
      \label{alg:up-edges}
    \begin{algorithmic}[1]
        \REQUIRE $v\in V$;
            $V_v$, array of all  vertices in $V$ above $v$,
            ordered clockwise;
            $in_v$, array of all incoming edges of~$v$,
            sorted radially clockwise;
            $\theta$, the minimum angle
            formed by any three vertices in~$\pi(V)$;
            and $\adgm$, the APD in direction~$e_2$
        \ENSURE array of all outgoing edges of $v$
        \STATE $\varindeg \gets$ indegree of $v$ in direction $-e_2$.
            \label{algln:up-edges:indeg}
        \STATE $\eAStack \leftarrow$ a stack of edge arc objects,
            initialized with a single edge arc $A$, where $A.\thevert=v$,
            $A.\alpha_1=0$,~$A.\alpha_2=\pi$,
            $A.\eC= \varindeg$, and $A.\vV=V_v$\label{algln:up-edges:stack}
        \STATE $E_v \gets \emptyset$\label{algln:up-edges:initEv}
        \WHILE{$\eAStack$ is not empty}\label{algln:up-edges:loop}
            \STATE $\eAVar\leftarrow \eAStack.pop()$\label{algln:up-edges:pop}
            \IF{$\eAVar.\eC = 0$} \label{algln:up-edges:zero-edges}
                \STATE \textbf{Continue} to top of while loop
				\label{algln:up-edges:zero-count}
            \ELSIF{$\eAVar.\eC = |\eAVar.\vV|$ }\label{algln:up-edges:one-arc}\label{algln:up-edges:zero-edges-end}
                \STATE Append $v \times \eAVar.\vV$ to $\eV$, in order
                    \label{algln:up-edges:addAll}
                \STATE \textbf{Continue} to top of while loop
		\label{algln:up-edges:found-edge}
            \ELSE\label{algln:up-edges:endone-arc}
                \STATE $(\eAVar_{\ell},\eAVar_r) \leftarrow
                    \splitArc{\eAVar}{in_v \cup \eV}{\theta}$\label{algln:up-edges:splitMain}
                \STATE Push $\eAVar_{r}$ onto $\eAStack$\label{algln:up-edges:push1}
                \STATE Push $\eAVar_{\ell}$ onto
		    $\eAStack$\label{algln:up-edges:push2}
            \ENDIF
        \ENDWHILE\label{algln:up-edges:loopend}
        \RETURN $E_v$
            \label{algln:up-edges:return}
    \end{algorithmic}
\end{algorithm}

%% file: body/algorithms/find-edges.tex
\begin{algorithm}[tbh]
    \caption{$\findEdges{V}$}
    \label{alg:find-edges}
    \begin{algorithmic}[1]
        \REQUIRE $V$, array of all vertices in the unknown graph
        \ENSURE $E$, array of all edges in the unknown graph
        \STATE $\adgm \gets \oracle{-e_2}{}$\label{algln:find-edges:getdgm}
        \STATE $E \gets \{\}$
        \STATE $\varcyclic \gets$ for each $v \in V$, an array clockwise ordering all
            vertices in $V$ that are above $v$\label{algln:find-edges:cyclic}
            \label{algln:find-edges:sort}
        \STATE $\theta \gets$ min angle defined by any three vertices of
            $\pi(V)$\label{algln:find-edges:theta}
        \FOR {$v$ in $V$, in increasing height in direction $e_2$}
          \label{algln:find-edges:callAlg2Begin}
        \STATE $\varincoming \gets$ clockwise sorted array of edges in $E$ incident to $v$
        \label{algln:find-edges:down-edges}
        \STATE $E += \findUpEdges{v}{\varcyclic[v]}{\varincoming}{\theta}{\adgm}$\label{algln:find-edges:up-edges}
        \label{algln:find-edges:callup-edges}
        \ENDFOR \label{algln:find-edges:callAlg2End}
        \RETURN $E$
    \end{algorithmic}
\end{algorithm}

%% file: body/discussion.tex
One way of proving that a
discretization of the APHT is faithful is through the method of \emph{reconstructing} the
underlying simplicial complex. That is, by showing
that the underlying simplicial complex can be recovered with the data of the
discretization alone.  In this paper, we take that approach and provide an
algorithm for reconstructing a graph immersed in $\R^d$. We use fewer
persistence diagrams than presented in alternate approaches.
For example,
the algorithm that we present for edge reconstruction (when the vertex locations
are known) uses~$\edgeDGMCompFull$ diagrams.
In contrast, \cite[Theorem 16]{belton2019reconstructing} uses $n^2 - n$
diagrams. Note that, for a very dense edge set, that is, when~$m = \Theta(n^2)$,
the method in~\cite[Theorem 16]{belton2019reconstructing} uses fewer diagrams.
However, if $m=O(n)$,  as is common in many complexes, the represntation
computed in this paper has fewer diagrams.  Moreover, we emphasize that the
number of diagrams is not exponential in the ambient dimension. 

One might hope to use binary search strategies to reconstruct a
simplicial complexe, but the methods presented here are unique to
one-skeletons. Radially ordering higher dimensional simplices is not well-defined, and
this issue prevents the
methods presented here from
being immediately transferrable.
On the other hand, with the
representation in this paper being output-sensitive (as opposed to testing if
every pair of vertices is a simplex), we have hope for the discretization of the
(A)PHT of a simplicial complex immersed in~$\R^d$ being proportional to the size of the
complex itself.

We also observe that not all diagrams used in our reconstruction algorithms were
strictly necessary (i.e., the set of diagrams used were not a minimal faithful
set). One straightforward way to reduce the number of diagrams used without
altering the method much would be to split the region above a vertex in the
sweep into arcs that contain exactly the same number of edges as vertices, or no
edges. This property can then be validated by a simple difference of indegrees.
In ongoing work, we hope to make these claims precise.
We also hope to extend our methods to use topological descriptors that
are \emph{not} dimension-returning (such as augmented Euler Characteristic
curves).

%% file: body/append-proofs.tex

In this appendix, we provide
the proofs omitted from \secref{edges}. These proofs provide justification for
the runtimes, diagram compelxity, and correctness of the
algorithms presented in this paper.

\subsection{Proof of \thmref{split-arc}}\label{append:split-arc-proof}

\splitarc*

\begin{proof}
    For the runtime, we walk through the algorithm and
    analyze the time and diagram complexity of each line.
    In \alglnrefRange{split:numverts}{split:alpha}, we find the angle
    $\alpha$ that splits $\eAVar.\vV$ into two equal sets,
    then in \alglnref{split:dir} compute a direction~$\dir$
    orthogonal to $\alpha$.
    See \figref{EA}.
    \alglnrefRange{split:numverts}{split:dir} use no
    diagrams and can be done in constant time when restricting our attention
    to the~$(e_1,e_2)$-plane.
    However, we need $\dir$ to be a
    direction in $\R^d$ (as opposed to only in the $(e_1,e_2)$-plane),
    so the computation takes~$\Theta(d)$ time.\footnote{With some clever data
    structures, this $\Theta(d)$ can be reduced to constant time.  For example,
    we could require vectors in $\R^2$ are automatically padded with $0$'s to
    become vectors in $\R^d$ when needed. However, this is out of the scope of
    the real RAM model of computation.}
    Specifically, $\dir$ is the vector
    \begin{align}\label{eqn:computeSplitDir}
        \dir &=e^{\frac{1}{2}i(2\alpha-\pi-\theta)} \\
        &= \Bigg( \cos{\left(\alpha-\frac{1}{2}\pi-\frac{1}{2}\theta\right)},
        \sin{\left(\alpha-\frac{1}{2}\pi-\frac{1}{2}\theta\right)}, 0, 0, \ldots,
        0 \Bigg) .\nonumber
    \end{align}
    To compute $m_{\ell}$ in \alglnref{split:countell}, we compute
    $\indeg{v}{\dir}$ then subtract the cardinality of the set~$S:=\{ b \in bigedges ~|~
    \measuredangle b < \pi+\alpha \}$, where $\measuredangle b$ is taken to mean
    the angle $b$ makes with the $e_1$-axis, when viewed as
    a vector with $\eAVar.\thevert$ as the origin.
    By \lemref{indegree}, we compute $\indeg{v}{\dir}$ via the oracle using \zeroindegDGMComp
    diagram and $\zeroindegTimeFull$ time.
    Since~$\varbigedges$ is sorted and since~$s$ lies in the $(e_1,e_2)$-plane,
    we can find the set $S$ in $\Theta(\log(|\varbigedges|))$ time.
    The subtraction in \alglnref{split:countell}
    takes constant time, as does \alglnref{split:countr}.

    In \alglnrefTwo{split:ell}{split:r}, we create two edge arc objects.  The
    time complexity of creating them is proportional to the size of the obejcts
    themselves.  All attributes of edge arc objects,
    except the array of vertices ($\vV$), are constant
    size.  By construction, $\eAVar_{\ell}.\vV$ and $\eAVar_r.\vV$
    split~$\eAVar.\vV$ into two sets, which can be done na\"ively
    in~$\Theta(d|\eAVar.\vV |)$ time by walking through~$\eAVar.\vV$ and storing
    each one explicitly.  However, we improve this to~$\Theta(\log|\eAVar.\vV
    |)$~time if we
    have a globally accessible array of vertices (sorted cw around~$v$) and just
    computes the
    pointers to the beginning and end of the sub-arrays corresponding to the
    $\vV$ attributes of the new edge arc objects.
    In total, \algref{split-arc}  and
    takes~$\Theta(d + \zeroindegTime+ \log(|\varbigedges|) + 1
    + \log(|\eAVar.\vV|)) = \splitTimeFull$ time and uses
    uses \splitDGMComp diagram.

    Now that we have walked through the algorithm and established the runtime
    and diagram complexity, we prove correctness. To do so, we first show that
    $\eAVar_{\ell}$ and $\eAVar_r$ are edge arc objects. In particular, this
    means showing that they have the correct values for $\eC$ and $\vV$.
    We prove this for $\eAVar_{\ell}$; the proof for $\eAVar_r$ follows a
    similar argument.

    \emph{$\eAVar_{\ell}.\eC$:}
    We must show that $\eAVar_{\ell}.\eC$ is the number of edges in
    $\eAVar_{\ell}$ incident to $\eAVar_{\ell}.\thevert$.
    By \lemref{count}, the value returned from~$\indeg{\eAVar.\thevert}{\dir}$ counts
    all edges incident to~$\eAVar.\thevert$ and below $\dir \cdot
    \eAVar.\thevert$ in direction~$\dir$. By
    \lemref{arccount}, this is
    exactly the total number of edges in
    $\eAVar_{\ell}$ plus edges~$(\eAVar.\thevert,v')\in \eV$ for which $\dir
    \cdot v'
    < \dir \cdot \eAVar.\thevert$.
    Thus, by subtracting $| \{ b \in bigedges ~|~ \measuredangle b < \pi+\alpha
    \}|$ from~$\indeg{\eAVar.\thevert}{\dir}$ on \alglnref{split:countell}, we
    are left with $m_\ell$, the number of edges incident to $\eAVar.\thevert$
    contained in $\eAVar_{\ell}$.  Setting~$\eAVar_{\ell}.\eC= m_{\ell}$ on
    \alglnref{split:ell}, we see that $\eAVar_{\ell}.\eC$ is correct.

    \emph{$\eAVar_{\ell}.\vV$:}
    We must show that $\eAVar_{\ell}.\vV$ contains an array of all verices
    contained in $\eAVar_{\ell}$ radially ordered clockwise. This follows from
    the fact that $\eAVar.\vV$ is all vertices contained in $\eAVar$ ordered
    clockwise, so when we restrict $\eAVar.\vV$ to $\eAVar.\vV[:mid]$ on
    \alglnref{split:ell}, we are eliminating vertices not contained in
    $\eAVar_{\ell}$, so $\eAVar_{\ell}.\vV$ is correct.

    Next, we prove \stmtref{split-concat}.  Recall that  $\eAVar.\vV$ orders the
    vertices in decreasing angle with $e_1$.  In \alglnref{split:alpha}, the
    angle~$\measuredangle \pi(\eAVar.\vV[mid]-\eAVar.\thevert)$ is the angle
    made by $\eAVar.\thevert$ with the middle vertex.  We tilt this angle by
    $\theta/2$ on \alglnref{split:alpha} to obtain angle~$\alpha$.  By
    construction of $\alpha$, $$ \measuredangle
    \pi(\eAVar.\vV[mid]-\eAVar.\thevert) > \alpha.  $$ By definition of
    $\theta$, the angle $\alpha$ satisfies:
    $$
    \alpha >
    \measuredangle \pi(\eAVar.\vV[mid+1]-\eAVar.\thevert).$$
    Since the array~$\eAVar.\vV$ is sorted,
    all vectors in the set~$\pi(\eAVar.\vV[:mid]-\eAVar.\thevert)$
    have an angle of
    at least $\alpha$ with~$e_1$ and all vectors
    in~$\pi(\eAVar.\vV[:mid]-\eAVar.\thevert)$  have an angle of at most $\alpha$.

    By \alglnrefRange{split:numverts}{split:calcmid} and
    \alglnref{split:countell}, we know that $\eAVar.\vV$ contains
    the first $m=\lceil{\frac{1}{2}|\eAVar.\vV|}\rceil$ vertices in
    $\eAVar.\vV$.  Hence, \stmtref{split-size-ell} holds.
    \stmtref{split-size-r} follows from
    \stmtrefTwo{split-concat}{split-size-ell}.
\end{proof}

\subsection{Proof of \thmref{up-edges}}\label{append:up-edges-proof}
\upedge*
\begin{proof}
    First, we analyze the time complexity of the algorithm and the number of
    diagrams it requires. By \lemref{indegree}, \alglnref{up-edges:indeg} can be
    computed in $\theta(\log n)$ time (since we are given the diagram and do not
    need an additional oracle query). Storing $A.\vV$ by storing a
    pointer to~$V_v$, we initialize~$\eAStack$  and $E_v$ in
    \alglnrefTwo{up-edges:stack}{up-edges:initEv}
    in constant time.

    To analyze the complexity of the loop in
    \alglnrefRange{up-edges:loop}{up-edges:loopend}, we first note that this is a
    radial binary multi-search.  When processing an edge arc, we decide whether
    all edges have been found or if we need to split the edge arc.
    If there is only one edge in the arc (i.e., $\eAVar.\eC=1$), then this loop is a binary
    search for an edge, using the angle with $e_1$ in the~$(e_1,e_2)$-plane as
    the search key.  When~$\eAVar.\eC>1$, we search for all edges, finding them
    in decreasing angle order (since arcs with larger angles are added after
    arcs of smaller angles). The else-if statement in
    \alglnrefRange{up-edges:one-arc}{up-edges:endone-arc} is where the edges are
    added to $\eV$.  Note that this shortcuts additional edge arc splitting by
    stopping the process once we find that the number of edges in the arc is
    equal to the number of potential vertices that can form the edges.  As a
    result, each edge above $v$ contributes to~$O(\log n)$ edge arcs being added to
    $\eAStack$ and, in the case that every other vertex is incident to an edge
    with~$v$, we have~$\Theta(\log n)$ edge arcs added to the stack.
    All operations in the while loop are constant time, except
    splitting the edge arc object in \alglnref{up-edges:splitMain},
    which uses \splitDGMComp
    diagram and takes~$\splitTimeFull$~time.

    The complexity of \algref{up-edges} is dominated by the complexity of
    the while loop: the algorithm
    uses~$\upDGMCompFull$ augmented persistence diagrams and takes
    takes~$\upTimeFull$~time.

    To prove correctness of this algorithm, we track how edges are stored.
    In this algorithm, $\eAStack$ is a stack of edge arc objects. We abuse
    notation slightly and say
    that an edge $(v', v) \in E$ is in $\eAStack$ if there exists an
    edge arc~$A \in \eAStack$ such that $\measuredangle(v',v) \in (A.\alpha_1,
    A.\alpha_2)$.  In this case, we also say $(v',v)$ is in $A$.  We claim the while loop
    has the following loop invariant: 
    \begin{enumerate}[(i)]
        \item every outgoing edge from $v$ is
        either in $\eAStack$ or is in the list of edges $E_v$;
            \label{li:contained}
        \item the edges of $E_v$ are radially sorted
        clockwise order;
            \label{li:evorder}
        \item the edge arcs in $\eAStack$ are non-overlapping and
        in radially clockwise order; and, 
            \label{li:stackorder}
        \item for all $(v',v) \in E_v$ and $(v'',v) \in \eAStack$, the angle
            $\measuredangle(v',v)$ is larger than the
            angle~$\measuredangle(v'',v)$.
            \label{li:betweenorder}
    \end{enumerate}
    Entering the loop,~$E_v$ is empty and $\eAStack$ is a single edge arc
    accounting for the entire half space above $v$. 
    Thus, all outgoing edges are in
    $\eAStack$ and $E_v$ is empty, so the loop invariant is initially true.
    Denote the edge arc that is processed in iteration $j$ by
    $\eAVar^j$. Suppose
    entering iteration $j$, the loop invariant holds.
    In \alglnref{up-edges:pop}, we pop
    $\eAVar^j$ from $\eAStack$. There are three cases to consider.
    First, if $\eAVar^j.\eC = 0$, there are no edges
    in $\eAVar^j$, so we enter the if clause in \alglnref{up-edges:zero-edges}
    and continue to the top of the while loop.
    Since no edges were removed from $\eAStack$ when~$\eAVar^j$ was popped, we
    do not add any edges to $E_v$ nor do we change any ordering of the
    previously sorted edges or edge arcs in $\eAStack$.
    Thus, the loop invariant is maintained.
    Second, if $\eAVar^j.\eC = |\eAVar^j.\vV|$, we enter the else-if clause in
    \alglnref{up-edges:one-arc}.
    Here, we know the vertices in the current edge arc are in a one-to-one
    correspondence with the edges
    in $\eAVar^j$, so we add these edges in sorted order to $E_v$ in
    \alglnref{up-edges:addAll}; since the loop iteration was true upon entering
    the loop, these edges are between the edges already in $E_v$ and the
    remaining edges in $\eAStack$. Thus, \liref{stackorder}
    and \liref{betweenorder} are maintained. 
    Furthermore, since we add
    the edges of $\eAVar^j$ in radially clockwise order,
    $E_v$ is still radially sorted in clockwise order, satisfying \liref{evorder}.
    Since all edges of $\eAVar^j$ are now in $E_v$, \liref{contained} is
    satisfied.
    So again, the loop invariant is maintained.
    The third and final case is that $\eAVar.\eC \neq 0$ and $\eAVar^j.\eC \neq
    |\eAVar^j.\vV|$. Then, we split the current edge arc on
    \alglnref{up-edges:splitMain}, pushing the two new halves onto the stack in
    a radially clockwise order, maintaining \liref{stackorder}.
    In this case, $E_v$ remains as it was going into the loop and the edges in $\eAStack$
    remain the same, so the loop invariant is maintained.
    To finalize the proof of partial correctness, suppose the loop ends and that
    the loop invariant is true. Since the loop ends, $\eAStack$ is empty. By
    \liref{contained} and \liref{evorder}, all outgoing edges are in $E_v$ in
    radially clockwise order.
    The runtime analysis tells us that the
    algorithm terminates, and thus, the algorithm is correct.
\end{proof}

\subsection{Proof of \thmref{edgerec}}\label{append:find-edges-proof}

\edgerec*
\begin{proof}
    We first analyze the runtime and diagram count for \algref{find-edges} by
    walking through the algorithm line-by-line.
    In \alglnref{find-edges:getdgm}, we ask the oracle for the diagram in
    direction $-e_2$, which takes $\oracleTimeFull$ time.
    In~\cite[Theorem~14 (Edge
    Reconstruction)]{belton2019reconstructing},
    simultaneously find the cyclic ordering around all vertices
    in~$\Theta(n^2)$ time by Lemmas~$1$ and~$2$ of \cite{verma2011slow}.  In
    \alglnref{find-edges:cyclic}, we do that as well; however, we do not store vertices
    that are above~$v$ in the array~$\varcyclic[v]$, and thus this line
    takes~$\Theta(n^2)$ time.  We note that such a cyclic ordering exists around
    each vertex by \assref{projectedindep}.
    Once we have $\varcyclic$, to find the minimum angle defined by any
    three vertices of $V$, we check all angles between vectors~$\varcyclic[v][i]-v$
    and~$\varcyclic[v][i+1]-v$ in \alglnref{find-edges:theta} in~$\Theta(n+m)$
    time.

    The for loop in
    \alglnrefRange{find-edges:callAlg2Begin}{find-edges:callAlg2End} is repeated
    $n$ times, once for each vertex in $V$.
    To determine the order of processing the vertices in $V$, we follow the
    births in $\adgm_0$, in decreasing order (since $\adgm_0$ is the lower-star
    filtration in direction~$-e_2$).  Thus, finding the order takes $\Theta(n)$
    time.
    In each iteration, we compute the incoming edges (those whose other vertex
    is below~$v$) in
    \alglnref{find-edges:down-edges} followed by all outgoing edges (those whose
    other vertex is above $v$) in
    \alglnref{find-edges:up-edges}.  By \assitm{unique}, every edge is
    either incoming or outgoing with respect to direction~$e_2$.  Thus, all
    edges incident to~$v$ are in $E$ once $E$ is updated in
    \alglnref{find-edges:down-edges}.
    By \thmref{up-edges}, when processing vertex $v$, the call to \algref{up-edges} on
    \alglnref{find-edges:callup-edges}
    takes~$\upTimeFull$ time and uses~$\upDGMCompFull$ diagrams.
    Summing over all vertices, we see that the loop
    in \alglnrefRange{find-edges:callAlg2Begin}{find-edges:callAlg2End}
    takes
    \begin{align*}
        \sum_{v \in V} \upTimeFull \\
        = \Theta(m\log n(\splitTime))
    \end{align*}
    time and uses~$\sum_{v \in V} \upDGMCompFull =
    \edgeDGMCompFull$ augmented persistence diagrams.

    In total, \algref{find-edges}
    takes $\oracleTimeFull + \Theta(n^2) + \Theta(n+m) + \Theta(n) + \Theta(m\log n(\splitTime))
    = \edgeTimeFull$ time
    and
    uses $\Theta(1) + \edgeDGMCompFull = \edgeDGMCompFull$ diagrams.

    Next, we prove the correctness of \algref{find-edges} (i.e., that all edges
    are found).
    In order to process vertices in order of their heights in the $e_2$ direction, we first sort
    them in
    \alglnref{find-edges:callAlg2Begin}. For $1 \le j \le n$, let $v_j$ be the $j$\th vertex
    in this ordering.
    To show that \algref{find-edges} finds all edges in $E$, we
    consider
    the loop invariant (LI): when we process~$v_j$, all edges
    with maximum vertex height equal or less than the height of~$v_j$ are
    known. The LI is trivially true for~$v_1$. We now assume that it is true for
    iteration $j$, and show that
    it must be true for iteration~$j+1$. By assumption, all edges~$(v_i, v_j)$
    with~$1 \leq i < j$ are known, and so by \thmref{up-edges},
    \algref{up-edges} finds all edges $(v_k, v_j)$, where~$k>j$,
    and we add them to the edge set~$E$. Note that, by
    assumption, all edges~$(v_x, v_i)$ for~$1 \leq i \leq j$ are also already
    known, and so the invariant is maintained. Thus, after the loop terminates,
    all edges are~found.
\end{proof}

%% file: body/append-basis.tex

In \assitm{unique}, we assume all vertices of the underlying graph are unique with respect to the
first basis direction~$e_2$.
In this appendix, we provide details of how to \emph{find} a basis
where all vertices have a unique height with respect to the second basis
direction. I.e., this appendix allows us to remove one general position
assumption by showing it can be satisfied deterministically, at an added cost
of~$\tiltTimeFull$ time.

\begin{lemma}[Creation of Orthonormal Basis] \label{lem:create_basis}
    Given a point set $P \subset \R^d$ satisfying~\assitm{affine} and \assitm{projectedindep},
    we can use two
    diagrams and $\tiltTimeFull$ time to create the orthonormal basis $\basis$ so
    that all points of~$P$ have a unique height in direction $b_2$.
\end{lemma}
\begin{proof}
    Algorithm 6 ($\algoName{Tilt}$) of~\cite{fasy2022faithful}
    takes diagrams from two linearly independent directions~$\dir, \dir' \in
    \S^{d-1}$, the point set~$P$, and returns a
    direction~$\dir_*$ in~$\tiltTimeFull$ time\footnote{While \cite{fasy2022faithful} does not
    account for diagram computation time, there are two diagrams used in this
    process, hence our addition of~$\oracleTimeFull$ to the total runtime.}
    so that the following
    properties holds for all~$p_1,p_2 \in P$:
    \begin{enumerate}[(i)]
	\item If $p_1$ is strictly above (below) $p_2$ with respect to direction $\dir$,
	    then~$p_1$ is strictly above (below, respectively) $p_2$ with
            respect to direction $\dir_*$.\label{stmt:tilt-first}
        \item If $p_1$ and $p_2$ are at the same height with respect to
            direction~$\dir$ and $p_1$ is strictly above (below) $p_2$ with respect to
            direction~$\dir'$, then $p_1$ is strictly above (respectively,
            below) $p_2$ with respect to
            direction $\dir_*$. \label{stmt:tilt-second}
        \item If $p_1$ is is at the same height as $p_2$ with respect to
            both directions $\dir$ and $\dir'$, then $p_1$ and $p_2$ are at the
            same height with respect to direction
            $\dir_*$.\label{stmt:tilt-third}
    \end{enumerate}
    A proof of correctness is given in~\cite[Lemma 32
    (Tilt)]{fasy2022faithful}.

    We start with the standard basis for $\R^d$, $\{ e_1, e_2, \ldots, e_d\}$,
    and we replace the first two basis elements as follows.
    Let~$b_2$ be the direction obtain by using
    $\algoName{Tilt}$ with~$s=e_1$,~$s'=e_2$, and~$P=P$.

    By \assitm{projectedindep},
    no three points of $P$ are colinear when projected onto the first two coordinates. In
    particular, this means no two points share the same heights in both
    the~$e_1$ and $e_2$ directions. Then, by \stmtsref{tilt-first}{tilt-second}
    above, the direction~$b_2$ must order all vertices of $P$ uniquely.
    Using only the first
    two coordinates of $b_2$ and $e_1$, we then perform
    Gram Schmidt orthanormalization to compute the first two coordinates of
    $b_1$.
    More precisely, letting~$b_i^{(j)}$ denote the $j$th coordinate of $b_i$, we compute
    \begin{equation}
        \begin{pmatrix} b_1^{(1)} \\ b_1^{(2)}\end{pmatrix} = \begin{pmatrix} 1
            \\ 0 \end{pmatrix} - \frac{\bigg\langle \begin{pmatrix} b_2^{(1)} &
                b_2^{(2)}\end{pmatrix}^T, \begin{pmatrix} 1 & 0
    \end{pmatrix}^T\bigg\rangle}{\bigg|\bigg|\begin{pmatrix} b_2^{(1)} &
    b_2^{(2)}\end{pmatrix}^T \bigg|\bigg|^2}
            \begin{pmatrix} b_2^{(1)} \\ b_2^{(2)}\end{pmatrix}
    \end{equation}
    We then set $b_1^{(j)}=0$ for $2 < j \leq d$, so that $b_1 \in \text{span}
    \{e_1, e_2 \}$, $b_2 \perp b_1$, and~$||b_1||=1$.  Only considering the first
    two coordinates of $b_2$ and $e_1$ means this process takes constant time.
    The remaining~$e_i$ for~$2 \leq i \leq d$ can
    be used to fill the basis.
    Finally, we have a basis satisfying all assumptions of \assref, namely,
    $\basis$.
\end{proof}